\documentclass{article}
\usepackage{amsmath}
\usepackage{amssymb}
\usepackage{amsfonts}
\usepackage{latexsym}
\usepackage{color}
\usepackage{graphicx}

\catcode `\@=11 \@addtoreset{equation}{section}

\catcode `\@=12

\newcommand{\be}{\begin{equation}}
\newcommand{\en}{\end{equation}}
\newcommand{\bea}{\begin{eqnarray}}
\newcommand{\ena}{\end{eqnarray}}
\newcommand{\beano}{\begin{eqnarray*}}
\newcommand{\enano}{\end{eqnarray*}}
\newcommand{\bee}{\begin{enumerate}}
\newcommand{\ene}{\end{enumerate}}

\newcommand{\mult}{\,{\scriptstyle \square}\,}
\newcommand{\vp}{\varphi}
\newcommand{\mc}{\mathcal}
\newcommand{\mb}{\mathbb}
\newcommand{\R}{{\mc R}}
\newcommand{\Q}{{\mc Q}}
\newcommand{\N}{\mathbb N}

\newcommand{\Hil}{{\cal H}}

\newcommand{\HH}{\mc H}

\newcommand{\F}{{\cal F}}
\newcommand{\B}{{\cal B}}
\newcommand{\Lc}{{\cal L}}
\newcommand{\LL}{{\cal L}}

\newcommand{\D}{{\cal D}}

\newcommand{\E}{{\cal E}}

\newcommand{\M}{{\cal M}}
\newcommand{\A}{\mathfrak A}
\newcommand{\Ao}{\A_0}
\newcommand{\MM}{{\mathfrak M}}
\newcommand{\DD}{{\mathfrak D}}
\def\P{{\mathcal P}}
\newcommand{\up}{\upharpoonright}

\newtheorem{thm}{Theorem}
\newtheorem{cor}[thm]{Corollary}

\newtheorem{prop}[thm]{Proposition}
\newtheorem{defn}[thm]{Definition}

\newtheorem{exa}[thm]{Example}
\newtheorem{remark}[thm]{Remark}
\newcommand{\beex}{\begin{exa}$\!\!${\bf }$\;$\rm }
\newcommand{\enex}{ \end{exa}}
\newenvironment{proof}{\noindent {\bf Proof:}}{\hfill$\Box$}
\newcommand{\rep}{\textsf{Rep}(\A,\Ao)}
\newcommand{\ip}[2]{\langle {#1}|{#2}\rangle}
\newcommand{\bedefi}{\begin{defn}$\!\!${\bf }$\;$\rm }
\newcommand{\findefi}{\end{defn}}
\newcommand{\berem}{\begin{remark}$\!\!\!${\bf} \rm }
\newcommand{\enrem}{ \end{remark}}
\newcommand{\w}{{\rm w}}
\def\b{{\textsf{b}}}

\begin{document}

\thispagestyle{empty}

\vspace*{1cm}

\begin{center}
{\Large \bf Representable states \\ on quasi-local quasi *-algebras}   \vspace{2cm}\\

{\large F. Bagarello}\\
  Dipartimento di Metodi e Modelli Matematici,
Fac. Ingegneria, Universit\`a di Palermo, I-90128  Palermo, Italy\\
e-mail: bagarell@unipa.it

\vspace{5mm}{\large C. Trapani}\\ Dipartimento di Matematica e Informatica, Universit\`a di Palermo,\\ I-90123 Palermo
(Italy)\\e-mail:
trapani@unipa.it

\vspace{5mm}
{\large S. Triolo}\\ Dipartimento di Matematica e Informatica, Universit\`a di Palermo,\\ I-90123 Palermo (Italy) \\e-mail:
salvo@math.unipa.it
\end{center}

\vspace*{1.5cm}

\begin{abstract}
\noindent
Continuing  a previous analysis originally motivated by physics, we consider representable states on quasi-local quasi *-algebras, starting with examining the possibility for a {\em compatible} family of {\em local} states  to give rise to a {\em global} state. Some properties of {\em local modifications} of representable states and some aspects of their asymptotic behavior are also considered.
\end{abstract}

\vspace{1cm}


\vfill

\newpage

% Section 1
\section{Introduction and Preliminaries}\label{sec_zero}
In a previous paper \cite{bit0} A. Inoue and two of us introduced a local structure for a locally convex quasi *-algebra obtained as the completion of a C*-normed algebra with respect to a weaker locally convex topology. The resulting object was called a {\em quasi-local quasi *-algebra}. The main motivation for this comes from the so-called algebraic approach to quantum systems with infinitely many degrees of freedom, once one recognizes that the original formulation of Haag and Kastler (see, e.g. \cite{sew}) in terms of C*-algebras is not well adapted to the mathematical description of many models (see \cite{fbrev, ctrev} for reviews on this subject and the references therein). The local structure is one of the basic points of the approach since it allows to study properties of the global system through the corresponding properties of its local parts. But often limits of local observables may fail to exist in the C*-norm of the local algebra and completions under weaker topologies (that, in general, fail to be *-algebras) must be taken. This puts on the stage {\em partial *-algebras} \cite{ait_book} and, more precisely, quasi *-algebras, originally introduced by Lassner for the study of the thermodynamical limits of certain spin systems \cite{lass1,lass2}.

In this paper we continue this study, focusing our attention on states (or, more generally, positive linear functionals) which allow a GNS-like representation and named, for this reason {\em representable}.
The paper is organized as follows. In Section \ref{sect_one}  we discuss first conditions for recovering from a {\em compatible} family of states on the the C*-algebras of the local net a {\em global} state. Then we show that the usual characterization of the purity of a state in terms of irriducibility of the corresponding GNS representation, well known in the case of C*-algebras, extends to  quasi-local quasi *-algebras.  In Section \ref{sect_two} we consider once more {\em local modifications} of representable states \cite{bit0} and give some more results on their asymptotic behavior.

\medskip To begin with, we recall briefly some definitions concerning quasi
*-algebras. More details can be found in \cite{ait_book}.

Let $\A$ be a linear space and $\Ao$ a  $^\ast$ -algebra contained
in $\A$. We say that $\A$ is a quasi  $^\ast$ -algebra with
distinguished  $^\ast$ -algebra $\Ao$ (or, simply, over $\Ao$) if
\begin{itemize}\item[(i)] the left multiplication $ax$ and the right multiplication $
xa$ of an element $a$ of $\A$ and an element $x$ of $\A_0$ which
 extend the multiplication of $\A_0$ are always defined and
bilinear; \item[(ii)] $x_1 (x_2 a)= (x_1x_2 )a$ and $x_1(a
 x_2)= (x_1 a) x_2$, for each $x_1, x_2 \in \A_0$ and $a \in \A$;

\item[(iii)] an involution $*$ which extends the involution of $\A_0$
is defined in $\A$ with the property $(ax)^*= x^*a^*$ and $(xa)^
* =a^* x^*$ for each $x \in \A_0$ and $a \in \A$.
\end{itemize}
A quasi  $^\ast$ -algebra $(\A,\Ao)$ is said to have a unit
$e$ if there exists an element $e \in \Ao$ such
that $a\,e =e\, a=a, \;\, \forall a\in \A$.

A quasi *-algebra $(\A, \Ao)$ is called locally convex if $\A$ is a locally convex space with topology $\tau$, with the properties:
(i) the involution $a \mapsto a^*$ is continuous in $\A$;
(ii) for every fixed $x \in \Ao$, the multiplications $a\mapsto ax$, $a\mapsto xa$ are continuous in $\A$;
(iii) $\Ao$ is dense in $\A$. A locally convex quasi*-algebra with topology $\tau$ is denoted by $(\A[\tau], \Ao)$.

Let $\D$ be a dense subspace of a Hilbert space $\Hil$. We denote by $ \Lc^\dagger(\D,\Hil) $ the set of all
(closable) linear operators $X$ such that $ {\D}(X) = {\D},\; {\D}(X^*) \supseteq {\D}.$

The set $ \LL^\dagger(\D,\Hil ) $ is a  partial *-algebra
 with respect to the following operations: the usual sum $X_1 + X_2 $,
the scalar multiplication $\lambda X$, the involution $ X \mapsto X^\dagger = X^* \up {\D}$ and the \emph{
(weak)} partial multiplication $X_1 \mult X_2 = {{X_1}^\dagger}^* X_2$, defined whenever $X_2$ is a weak right
multiplier of $X_1$ (we shall write $X_2 \in R^{\rm w}(X_1)$ or $X_1 \in L^{\rm w}(X_2)$), that is, iff $ X_2
{\D} \subset {\D}({{X_1}^\dagger}^*)$ and  $ X_1^* {\D} \subset {\D}(X_2^*).$

If $\MM\subseteq \LL^\dagger(\D,\Hil ) $ we denote by $\MM_\b$ its {\em bounded part}; i.e, $\MM_\b:=\{ X\in \MM;\, X \mbox{ is a bounded operator on }\D\}.$

Let $\Lc^\dagger(\D)$ be the subspace of $\Lc^\dagger(\D,\Hil)$ consisting of all its elements  which leave, together with their adjoints, the domain $\D$ invariant. Then $\Lc^\dagger(\D)$ is a *-algebra with
respect to the usual operations.\\
It is easy to see that $(\LL^\dagger(\D,\Hil ), \Lc^\dagger(\D)_\b)$ is a quasi *-algebra.

If $\mathcal O\subset\Lc^\dagger(\D,\Hil)$, with $\mathcal O=\mathcal O^\dagger$, then the weak bounded commutant $\mathcal O_\w'$ is defined by
$$
\mathcal O_\w'=\{X\in\Lc_\b^\dagger(\D,\Hil): \ip{A\xi}{X^\dagger \eta}= \ip{X\xi}{A^\dagger \eta}, \, \forall A\in \mathcal O, \,\xi,\eta\in\D\}.
$$

Let $(\A,\Ao)$ be a quasi *-algebra with identity $e$ and $\D_\pi$  a dense domain in a certain Hilbert
space $\Hil_\pi$. A linear map $\pi$ from $\A$ into $\LL^\dagger(\D_\pi, \Hil_\pi)$ such that:

(i) $\pi(a^*)=\pi(a)^\dagger, \quad \forall a\in \A$,

(ii) if $a\in \A$, $x\in \Ao$, then $\pi(a)${$\Box$}\!\! $\pi(x)$ is well defined and
$\pi(ax)=\pi(a)${$\Box$}\!\! $\pi(x)$,

\noindent
is called  a *-representation of $\A$. Moreover, if

(iii) $\pi(\Ao)\subset \LL^\dagger(\D_\pi)$,

\noindent then $\pi$ is said to be a *-representation of the quasi *-algebra $(\A,\Ao)$.

\medskip If $\pi$ is a *-representation of $(\A, \Ao)$, then the {\em closure} $\widetilde{\pi}$ of $\pi$ is defined, for each
$x \in \A$, as the restriction of $\overline{\pi(x)}$ to the domain $\widetilde{\D_\pi}$, which is the completion
of $\D_\pi$ under the {\em graph topology} $t_\pi$ \cite{ait_book} defined by the seminorms $\xi \in \D_\pi \to
\|\pi(a)\xi\|$, $a\in \A$. If $\pi=\widetilde{\pi}$ the representation is said to be {\em closed}.

The representation $\pi$ is said to be {\em ultra-cyclic} if there exists $\xi_0\in\D_\pi$ such that
$\D_\pi=\pi(\Ao)\xi_0$, while is said to be {\em cyclic} if there exists $\xi_0\in\D_\pi$ such that
$\pi(\Ao)\xi_0$ is dense in $\D_\pi$ w.r.t. $t_\pi$.

Let $(\A, \Ao)$ be a quasi-*algebra. We denote by $\rep$ the family of all *-representations of $(\A,\Ao)$.

\section{Quasi local structure and representations}\label{sect_one}

We consider now the  case where $\A$ has a {\em local structure}. Following \cite{brarob} we construct the local net of C*-algebras as
follows.

Let $\F$ be a set of indexes directed upward and with an orthonormality relation $\perp$ such that

\begin{itemize} \item[(i.)] $\forall
\alpha\in\F$ there exists $\beta\in\F$ such that $\alpha\perp\beta$; \item[(ii.)] if $\alpha\leq \beta$ and $\beta\perp
\gamma$,\, $\alpha, \beta, \gamma\in\F$, then $\alpha\perp\gamma$; \item[(iii.)] if, for $\alpha, \beta, \gamma\in\F$,\,
$\alpha\perp\beta$ and $\alpha\perp\gamma$, there exists $\delta\in\F$ such that $\alpha\perp\delta$ and
$\delta\geq \beta, \gamma$.
\end{itemize}

\bedefi \label{localstruct} Let $(\A[\tau],\Ao)$ be a locally convex quasi *-algebra with unit $e$. We say that $(\A, \Ao)$ has a {\em local structure} if there exists a net
$\{\A_\alpha(\|.\|_\alpha), \,\alpha\in\F\}$  of subspaces of $\Ao$, indexed by $\F$, such that every $\A_\alpha$ is a C*-algebra (with norm $\|.\|_\alpha$ and unit $e$) with the properties

\begin{itemize} \item[(a)] $\Ao=\bigcup_{\alpha \in \F} \A_\alpha$
\item[ (b)] if $\alpha\geq
\beta$ then $\A_\alpha\supset\A_\beta$; \item[(c)] if
$\alpha\perp\beta$, then $xy=yx$ for every $x\in\A_\alpha$, $y\in\A_\beta$.
 \item[(d)] if $x \in \A_\alpha \cap \A_\beta$, then $\|x\|_\alpha =\|x \|_\beta$.
\end{itemize}
A quasi *-algebra $(\A[\tau], \Ao)$ with a local structure will be shortly called a {\em quasi-local quasi*-algebra}.
\findefi

\berem There is no loss of generality in assuming that $\A[\tau]$ is a complete locally convex space, since taking the completion would maintain untouched the local structure. \enrem

If $(\A[\tau], \Ao)$ is a quasi-local quasi-*algebra, and $x \in \Ao$, then there exists $\beta\in\F$ such that $x\in\A_\beta$.  We put
$J_x= \{{\alpha \in \F: x \in \A_\alpha\}}$. Because of (d), $\|x\|:=\|x\|_\alpha$, $\alpha \in J_x$, is  a well-defined norm on $\Ao$.
Then $\Ao$ is a C*-normed algebra with norm $\| \cdot\|$.  Moreover, putting $\A_\infty:=\cap_{\alpha\in\F}\A_\alpha$, we will assume that $\forall\,
x\in\A_0$, $x\notin \A_\infty$, there exists $\alpha_x\in\F$ such that $\cap_{\beta\in
J_x}\A_\beta=\A_{\alpha_x}$.  We call $\alpha_x$ the {\em  support of $x$}, \cite{bit0}.

\berem First we notice that if $(\A, \Ao)$ is a quasi-local quasi *-algebra and $\pi \in \rep$, then $\pi_\alpha:=\pi\upharpoonright_{\A_\alpha}$ is bounded (being a *-representation of a C*-algebra). Moreover, as already stated, if $x\in \Ao$, then $x$ belongs to some $\A_\alpha$, $\alpha \in \F$. Then, we get
\begin{equation}\label{3.1}\|\pi(x)\|=\|\pi_\alpha (x)\|\leq \|x\|_\alpha =\|x\|,\quad \forall x \in \Ao.\end{equation}
Hence $\pi$ is bounded on $\Ao$, too.
\enrem

The following proposition, originally given in \cite{ct_ban}, extends the GNS construction to quasi
*-algebras.

\begin{prop}

Let $\omega$ be a linear functional on $\A$ satisfying the
following requirements:

(L1) $\omega(a^*a)\geq 0$ for all $a\in\Ao$;

(L2) $\omega(b^*x^*a)=\overline{\omega(a^*xb)}$,
$\forall\,a,b\in\Ao$, $x\in\A$;

(L3) $\forall x\in\A$ there exists $\gamma_x>0$ such that
$|\omega(x^*a)|\leq \gamma_x\,\omega(a^*a)^{1/2}$.

Then there exists a triple $(\pi_{\omega}, \lambda_{\omega},
\Hil_{\omega})$ such that

\begin{itemize}
  \item $\pi_{\omega}$ is a ultra-cyclic *-representation of $\A$ with ultra-cyclic vector $\xi_\omega$;
  \item $\lambda_{\omega}$ is a linear map of $\A$ into
  $\Hil_{\omega}$ with $\lambda_{\omega}(\Ao)=\D_{\pi_\omega}$, $\xi_\omega=\lambda_{\omega}(e)$ and
  $\pi_{\omega}(x)\lambda_{\omega}(a)=\lambda_{\omega}(xa)$, for every $x \in\A,\, a \in \Ao$.
  \item $\omega(x)=\ip{\pi_{\omega}(x)\xi_\omega}{\xi_\omega}$,
  for every $x \in \A$.
\end{itemize}

\end{prop}

\bedefi\label{defn_representable}A linear functional satisfying (L1),(L2),(L3) is called \emph{representable}. We denote by $\R(\A)$ the set of representable linear functionals on $\A$.\findefi

It is easily seen that $\R(\A)$ is a cone (i.e. if $\omega_1, \omega_2 \in \R(\A)$, then $\omega_1+\omega_2 \in \R(\A)$ and $\lambda \omega_1 \in \R(\A)$, for $\lambda \geq 0$). Before going forth we give some easy examples.

\beex As we mentioned in Section \ref{sec_zero}, if $\D$ is a dense subspace of Hilbert space $\Hil$, $(\LL^\dagger(\D,\Hil ), \Lc^\dagger(\D)_\b)$ is a quasi *-algebra. For every $\xi \in \D$, the linear functional $\omega_\xi$ defined by
$$ \omega_\xi (X)= \ip{X\xi}{\xi}, \quad X \in \LL^\dagger(\D,\Hil )$$ is representable.
\enex
\beex Consider the quasi *-algebra $(L^p(I), L^\infty(I))$, where $L^p(I)$  ($p\geq 1$) and $L^\infty(I)$ are the usual Lebesgue spaces on $I:=[0,1]$.
Put $\omega(f)=\int_0^1 f(x)dx$, $f \in L^p(I)$. If $1\leq p <2$, $\omega$ is not representable. Indeed, if $f\in L^p(I)\setminus L^2 (I)$, there cannot exist any $\gamma_f>0$ such that
$$ |\omega(f\vp)|= \left|\int_0^1 f(x)\vp(x) dx\right|\leq \gamma_f\, \omega(\vp^*\vp)^{1/2}= \gamma_f\,\|\vp\|_2, \quad \forall \vp \in L^\infty(I),$$ since this would imply that $f \in L^2(I)$. \enex

Let now $(\A[\tau], \Ao)$ be a { quasi-local quasi *-algebra}. If $\omega$ is a linear functional on $\A$, we put $\omega_\alpha=\omega \upharpoonright_{\A_\alpha}$. If $\omega$ is representable, then $\omega_\alpha$ is a positive linear functional on $\A_\alpha$.
Let $\{\pi_\omega,\HH_\omega,\xi_\omega\} $ be the GNS construction corresponding to $\omega$. Then, $[\pi_{\omega}(\A_0)\xi_\omega]={\HH}_{\omega}$ , where $[\M]$ denotes the closure of the subspace $\M$ of $\HH$.  The restriction $\pi_{\omega}
\upharpoonright_{\A_\alpha} $ of $\pi_\omega$ to $\A_\alpha$ is a
*-representation of $\A_\alpha$ in $\HH_\alpha :=[\pi_{\omega}(\A_\alpha)\xi_\omega]$.
Then  $\pi_{\omega} \upharpoonright_{\A_\alpha}$ is unitarily equivalent to the GNS representation $\pi_{\omega_\alpha}$ of $\A_\alpha$ induced by $\omega_\alpha$.

% il vecchio remark 5 \`e in cosa all'articolo.

\medskip Let us now suppose that, for every $\alpha \in \F$, we are given a positive linear functional $\omega_\alpha$. A natural question then arises: Does there exist a representable linear functional $\omega$ on $\A$ such that $\omega_\alpha=\omega \upharpoonright_{\A_\alpha}$, for every $\alpha \in \F$? A minimal requirement for a positive answer is that the family $\{\omega_\alpha;\, \alpha \in \F\}$, satisfies a {\em compatibility} condition like the following one:
\begin{equation} \label{coherence} \omega_\alpha (x) = \omega_\beta (x), \quad \forall x \in \A_\alpha \cap \A_\beta. \end{equation}
In this case we can define, for $x \in \Ao$,
\begin{equation}\label{defn_omega}  \omega(x) =\omega_\alpha (x), \quad \alpha \in J_x.\end{equation}
It is clear that $\omega$ is positive on $\Ao$ and that $\omega(x)\leq \omega(e) \|x\|$, for every $x \in \Ao$.
Our next goal is now to extend $\omega$, at least, to some elements of $\A\setminus \Ao$. This is, in general, a quite difficult problem to handle. We refer to two recent papers by Bellomonte and two of us \cite{btt1, btt2} for a more complete discussion. An obvious answer to this question can certainly be given when $\omega$ is $\tau$-continuous. But this is a rather restrictive assumption in concrete situations.

A possible approach \cite{btt2} consists in assuming that the sesquilinear form $\Omega$ associated to $\omega$
$$ \Omega(x,y) = \omega(y^* x), \quad x,y \in \Ao$$
is {\em closable}. This means that
\begin{itemize}
\item[(\textsf{cl})] if $x_n\stackrel{\tau}{\to} 0$, $x_n\in \Ao$ and $\Omega(x_n-x_m, x_n-x_m)\to 0$, then $\Omega(x_n, x_n)\to 0$.
\end{itemize}
In this case $\Omega$ has a {\em closure} $\overline{\Omega}$ defined on $\D(\overline{\Omega})\times \D(\overline{\Omega})$, where
$$\D(\overline{\Omega})=\{x \in \A; \exists x_n \stackrel{\tau}{\to} x, \, x_n\in \Ao,\, \Omega(x_n-x_m, x_n-x_m)\to 0\}$$ by
$$ \overline{\Omega}(x,y)= \lim_{n \to \infty} \Omega(x_n, y_n), \quad x,y \in \D(\overline{\Omega}).$$
The set $\D(\overline{\Omega})$ is, in general, only a vector subspace of $\A$, containing $\Ao$.

\smallskip An extension $\overline{\omega}$ of $\omega$ is now easy to define:
$$ \overline{\omega}(x) = \overline{\Omega}(x, e), \quad x \in \D(\overline{\Omega}).$$

\smallskip Now assume that $\Omega$ enjoys, for every sequence $\{x_n\}\subset \Ao$, the following property
\begin{itemize}
\item[(\textsf{wt})] $ \lim_{n \to \infty} \Omega(x_n,x_n)= 0 \; \Leftrightarrow \; \lim_{n \to \infty} \Omega(x_n^*,x_n^*)=0$.
\end{itemize}
In this case, if $x \in  \D(\overline{\Omega})$, then also $x^* \in \D(\overline{\Omega})$.
Moreover, taking into account that if $a,x \in \Ao$ there exists $\alpha \in \F$ such that $a,x \in \A_\alpha$, from a well-known property of positive linear functionals on C*-algebras, we have
\begin{eqnarray*}
\Omega(ax,ax) &=& \omega(x^*a^*ax)= \omega_\alpha(x^*a^*ax)\\ &\leq& \|a\|^2\omega_\alpha(x^*x)=\|a\|^2\Omega(x,x).
\end{eqnarray*}
If $x \in \D(\overline{\Omega})$ the previous inequality extends by a limit argument.
This inequality easily implies that if $a \in \Ao$ and $x \in \D(\overline{\Omega})$, then $ax \in \D(\overline{\Omega})$. Since, $\D(\overline{\Omega})$ is stable under involution we also get $xa \in \D(\overline{\Omega})$. Hence $(\D(\overline{\Omega}), \Ao)$ is a quasi *-algebra, with $\D(\overline{\Omega})$ dense in $\A$.
It is now easy to check, by simple limit arguments, that $\overline\omega$ satisfies (L1), (L2) and (L3) when considered in the quasi *-algebra $(\D(\overline{\Omega}), \Ao)$.

\bedefi Let $(\A[\tau], \Ao)$ be a quasi-local quasi-*algebra.  Let $\DD(\omega)$ be a subspace of $\A$ containing $\Ao$ and $\omega$ a linear functional defined on $\DD(\omega)$. We say that $\omega$ is \emph{quasi-representable} if $(\DD(\omega), \Ao)$ is a quasi *-subalgebra of $(\A,\Ao)$ and $\omega$ is representable on $(\DD(\omega), \Ao)$. If $\DD(\omega)=\A$ we say, simply, that $\omega$ is \emph{representable} \underline{on} $\A$.
We denote, respectively,  by $\Q(\A)$ and $\R(\A)$ the set of quasi-representable and representable linear functionals on $\A$.\findefi

The previous discussion can be summarized by the following
\begin{prop} Let $(\A[\tau], \Ao)$ be a quasi-local quasi-*algebra with net of C*-algebras $\{\A_\alpha;\, \alpha \in \F\}$. Let, for every $\alpha \in \F$, $\omega_\alpha$ be a positive linear functional on $\A_\alpha$ and assume that the family $\{\omega_\alpha; \, \alpha \in \F\}$ satisfies \eqref{coherence}. Define $\omega$ on $\Ao$ by \eqref{defn_omega}. Then, if the corresponding positive sesquilinear form $\Omega$ satisfies the conditions {\em (\textsf{cl})} and {\em (\textsf{wt})}, $\omega$ has a quasi-representable extension.

\end{prop}
\berem It is clear that if $\omega$, defined on $\DD(\omega)$, is quasi-representable, then it is representable as a linear functional on $\DD(\omega)$. Thus, also in this case, the restriction of $\omega$ to every $\A_\alpha$ is a positive linear functional on $\A_\alpha$.
\enrem

The notion of representability of a linear functional on the quasi local quasi *-algebra $(\A[\tau], \Ao)$ is independent of the topology $\tau$ given on $\A$. The latter has not been specified so far, leaving us some freedom for a reasonable choice of $\tau$, in order to make the interplay between representations and topology as much regular as possible, as it happens for C*-algebras. For this we need some additional assumption.

\bigskip
Let $(\A,\Ao)$ be a quasi *-algebra such that $\rep$ is {\em faithful} in the following sense:
$$\forall x \in \A,\, x\neq 0 \; \mbox{there exists } \pi \in \rep \mbox{ such that } \pi(x)\neq 0.$$
Then we can endow $\A$ with the topology $\tau_{rep}$ which is the weakest locally convex topology on $\A$ such that every $\pi\in \rep$ is continuous from $\A$ into $\Lc^\dagger(\D_\pi,\Hil_\pi)[\tau_{s^*}]$, where $\tau_{s^*}$ stands for the {\em strong* topology} of $\Lc^\dagger(\D_\pi,\Hil_\pi)$.
The topology $\tau_{rep}$ of $\A$ is then defined by the set of seminorms
$$ a\in \A \to \|\pi(a)\xi\|+ \|\pi(a^*)\xi\|, \quad \pi \in \rep,\, \xi \in \D_\pi.$$
By (\ref{3.1}), the topology induced by $\tau_{rep}$ on $\Ao$ is weaker than the norm $\|\cdot\|$.

\berem The fact that $\rep$ is faithful, easily implies that $(\A,\Ao)$ has a {\em sufficient} family of linear functionals satisfying (L1)-(L3).
Indeed, if $\pi\in \rep$ and $\xi \in \D_\pi$, then the linear functional $\omega^\pi_\xi$ defined by
$$ \omega^\pi_\xi (x)=\ip{\pi(x)\xi}{\xi} $$
satisfies (L1)-(L3). Were $\omega^\pi_\xi (x)=0$ for every $\pi\in \rep$ and $\xi \in \D_\pi$, then we would have $\pi(x)=0$, for every $\pi\in \rep$, which, in turn, implies $x=0$.
\enrem

From now on we will suppose that $\Ao$ is {\em dense} in
$\A[\tau_{rep}]$ (were it not so, we could replace $\A$ with the
closure of $\Ao$).

\begin{prop} The following statements hold.
\begin{itemize}\item[(i)] $(\A[\tau_{rep}],\Ao)$ is a locally convex quasi *-algebra.
\item[(ii)] Every $\omega \in \R(\A)$ is $\tau_{rep}$-continuous.
\end{itemize}
\end{prop}
\begin{proof} (i): We need to prove that multiplications are separately continuous. Let $a_\lambda \stackrel{\tau_{rep}}{\to} a$ and $x \in \Ao$. Then for every $\pi \in \rep$ and $\xi \in \D_\pi$ we have
$$ \|\pi (a_\lambda-a)\xi\|\to 0 \quad \mbox{and}\quad \|\pi (a_\lambda^*-a^*)\xi\|\to 0.$$
Since $\pi(x)\xi \in \D_\pi$ we also get $\|\pi (a_\lambda-a)\pi(x)\xi\|\to
0$ and from the boundedness of $\pi(x)^*$ we also have
$$ \|\pi(x^*)\pi (a_\lambda^*-a^*)\xi\|\leq \|x\|\|\pi (a_\lambda^*-a^*)\xi\|\to 0.$$ Hence,
$$ \|\pi (a_\lambda-a)\pi(x)\xi\|+\|\pi(x^*)\pi (a_\lambda^*-a^*)\xi\| \to 0.$$
The proof of the continuity of the left multiplication is similar.

(ii): Let $\pi_\omega$ be the GNS representation of $\omega$ and $\xi_\omega$ the corresponding cyclic vector. Then $\pi_\omega$ is continuous. Hence, we have
\begin{eqnarray*}
|\omega(x)|&=& |\ip{\pi_\omega(x)\xi_\omega}{ \xi_\omega}| \\
&\leq & \|\pi_\omega (x)\xi_\omega\|\|\xi_\omega\| \\
&\leq & ( \|\pi_\omega (x)\xi_\omega\|+ \|\pi_\omega (x^*)\xi_\omega \parallel )\|\xi_\omega\|\\
&=& \omega(e) ( \|\pi_\omega (x)\xi_\omega\|+ \|\pi_\omega
(x^*)\xi_\omega\parallel),
\end{eqnarray*} which proves the statement.
\end{proof}

%\begin{prop} $(\A, \Ao)$ is a locally convex quasi C*-normed algebra.
%\end{prop}
%\begin{proof} It is easily seen that $\tau_{rep}$ and $\|\cdot\|$ are compatible and that $(\A, \Ao)$ satisfies the conditions (T1)-(T3) of \cite{bfit_cstarnormed}. (CHECK)
%\end{proof}

%\berem If $\pi \in\rep$, then for every $\xi \in \D_\pi$ we can define a se is faithful.  \enrem

%From now on $(\A,\Ao)$ will always denote a  quasi
%*-algebra where $\A$ is endowed with $\tau_{rep}$ and $\Ao$ is
%dense in $\A$.

\begin{prop}\label{nove} Let $(\A[\tau_{rep}], \Ao)$ be a quasi-local quasi *-algebra and $\pi \in \rep$. Then $\pi(\A)'_\w = \pi(\Ao)'_\w$ and $\pi(\A)'_\w$ is a von Neumann algebra.\label{trap}
\end{prop}
\begin{proof}
Of course, $\pi(\Ao)\subseteq\pi(\A)$. Then, $\pi(\A)'_\w \subseteq
\pi(\Ao)'_\w.$ If $B\in\pi(\Ao)'_\w$ we have
$$\ip{B\pi(x_o)\xi}{\eta}=\ip{B\xi}{\pi(x{_o}^*)\eta}, \quad \forall x_{o}\in {\Ao},\, \xi, \eta\in \D_\pi.$$
For every $x\in\A$, there exist a net $x_\lambda\in\Ao$ such
that $x_\lambda \stackrel{\tau_{rep}}{\to} x$. Therefore
$$\parallel \pi(x_\lambda-x)\xi \parallel +\parallel  \pi(x_\lambda^*-x^*)\xi \parallel \rightarrow 0, \,\quad \xi\in\D_\pi$$
then
$$|\ip{B(\pi(x_\lambda)-\pi(x))\xi}{\eta}| \leq
\parallel(\pi(x_\lambda)-\pi(x))\xi \parallel \parallel B^*\eta \parallel \rightarrow 0$$
Similarly, we can prove that $|\ip{B\xi}{\pi(x_\lambda^*-x^*)\eta}|\rightarrow
0 $. Hence, $$ \ip{B\pi(x)\xi}{\eta}=\lim_\lambda
\ip{B\pi(x_\lambda )\xi}{\eta}=\lim_\lambda
\ip{B\xi}{\pi(x_\lambda^*)\eta}=\ip{B\xi}{\pi(x^*)\eta};
$$ i.e., $B\in\pi(\A)'_\w.$

Finally, $\pi(\A)'_\w$ is a von Neumann algebra, since $\pi(\Ao)'_\w=\pi(\Ao)'$, being $\pi(\A_0)$ a *-algebra of bounded operators.
\end{proof}

\medskip

It is well known that for states on a C*-algebra there is equivalence between extremality, purity and irreducibility of the corresponding GNS representations. We will now consider the relationship between the analogous statements in the case of a quasi-local quasi *-algebra. This analysis will involve an ordering of representable linear functionals based on positive elements of $\A$ and a convenient extension of the notion of irreducibility of a *-representation. Both notions will be defined in the sequel. We first fix the  following notations for subsets of $\R(\A)$ that will play a role in our discussion.
\begin{itemize}
%\item $\R(\A)$ is vector space;
\item $\B({\A}):=\{\omega \in \R(\A), \,\omega(e)\leq 1 \}$ is a convex subset of $\R(A)$;
 \item $\E({\A}):=\{\omega\in \B({\A}): \, \omega(e)=1 \}$ is a convex subset of $\B(A)$.
\end{itemize}
Elements of $\E({\A})$ will be called {\em states} of $(\A[\tau_{rep}],\Ao)$. As for C*-algebras, if $\omega$ is an extremal point of $\B({\A})$, then $\omega \in \E({\A}).$

\bedefi A *-representation $\pi$ of $(\A,\Ao)$ is said to be {\em quasi-irreducible}
if the set
 $\pi(\A)'_\w$ consists of multiples of the identity operator.
\findefi

\berem When dealing with a representation $\pi$ of a C*-algebra ${\mathfrak B}$, {\em irreducibility} means, loosely speaking, that $\pi$ cannot be decomposed into the direct sum of two nontrivial sub *-representations or, equivalently, that $\pi$ has no nontrivial invariant or {\em reducing} subspaces. These two statements in turn are equivalent to the usual commutant $\pi({\mathfrak B})'$ to consist only of multiples of the identity operator. For representations of general *-algebras the corresponding statements are mostly nonequivalent (invariant and reducing subspaces are different notions; a projection picked in the weak commutant may not define a sub *-representation, etc.). All these topics have been extensively discussed in the literature, starting from the pioneering paper by Powers \cite{powers} and have been, we may say, definitely systematized  in Schm\"udgen monograph \cite[Ch.8]{schmud} to which we refer for full details. The condition  $\pi(\A)'_\w={\mb C}I$ for studying irreducibility was first used in Power's paper, but in fact in the case he considered (self-adjoint *-representations) the weak commutant coincides with the commutant $\pi(\A)'_{\rm ss}$ considered by Schm\"udgen.
The fact that $\pi(\A)'_\w$ is, for a quasi-local quasi *-algebra, a von Neumann algebra makes the notion of quasi-irreducibility convenient for the purposes of this paper. We leave to future papers the study of its interplay with invariant or quasi-invariant \cite{bit0} subspaces.
\enrem

\medskip
Now we turn to the order structure. Let 
$$\Ao^+:= \left\{ \sum_{k=1}^n x_k^* x_k; \; x_k \in \Ao, \, n \in {\N}\right \}.$$
It is clear that if $y \in \Ao^+$, then for every $\omega \in \R(\A)$, $\omega(y)\geq 0$, as immediate consequence of (L1) and the linearity of $\omega$. For analogous reasons, we also have $\pi(y)\geq 0$ for every $\pi \in \rep$, by which we mean that $\ip{\pi(y)\xi}{\xi}\geq 0$, for every $\xi \in \D_\pi$.
\begin{prop} Consider the following statements.
\begin{itemize}
\item[(i)] $x \in \overline{\Ao^+}^{\tau_{rep}}$;
\item[(ii)] $\pi(x) \geq 0$, for every $\pi \in \rep$.
\item[(iii)] $\omega(x)\geq 0$, for every $\omega \in \R(\A)$.
\end{itemize}
Then, we have $$ (i) \Rightarrow (ii) \Leftrightarrow (iii) $$
\end{prop}

\begin{proof}

  (i)$\Rightarrow$ (ii):
Let $\pi\in\rep$ and   $x \in
\overline{\Ao^+}^{\tau_{rep}}$. Then, there exist a net $\{x_\lambda\}\subset\Ao^+$ such that
$x_\lambda \stackrel{\tau_{rep}}{\to} x$ then
$\ip{\pi(x)\xi}{\xi}=\lim \ip{\pi(x_\lambda) \xi }{\xi}\geq0$.

(ii)$\Rightarrow$ (iii) Let $\omega \in \R(\A)$ and $\{\pi_\omega,\HH_\omega,\xi_\omega\} $ the corresponding GNS construction. Then $\omega(x)= \ip{\pi_\omega(x)\xi_\omega}{\xi_\omega}\geq 0$, by the assumption.

(iii) $\Rightarrow$ (ii) is trivial.

\end{proof}

We leave open the question as to whether (ii)$ \Rightarrow$ (i). The previous Proposition, however,  suggests the definition $\A^+:=\overline{\Ao^+}^{\tau_{rep}}$.

This allows us to introduce an order in $\rep$: we say that $\nu<\omega$ if $(\omega-\nu)(a)\geq 0$  for every
$\,\omega, \nu\in \R(\A), a\in\A^+.$

\bedefi A state $\omega\in \E(\A)$ is called {\em pure}  if, for every  $\nu\in
\B(\A)$, $ \mbox{ $0\leq\nu\leq\omega$},$ there exists $ \lambda\in [0,1]$, such that
$\nu=\lambda\omega.$
We denote by
$\P({\A})$ the set of pure states.
\findefi

Next we prove that, similarly to C*-algebras, the
notions of purity of a state $\omega$ and (quasi)-irreducibility of the
representation associated with $\omega$ are intimately related.

\begin{thm}
Let $(\A,\Ao)$ be a {\em quasi-local quasi *-algebra} and
$\omega$ a state over $\A$ such that $\pi_{\omega}\in \rep$. The following statements are
equivalent.
\begin{itemize}
\item[$(i)$] $\omega$ is an extremal point of $ \B(\A)$.
\item[$(ii)$] $\omega$ is a pure state $(\omega\in \E(\A))$.
\item[$(iii)$] $\pi_\omega(\A)'_\w={\mb C}I$; i.e. $\pi_\omega$ is quasi-irreducible.
\end{itemize}

\end{thm}
\begin{proof}
The equivalence of ${\rm (ii)}$ and  ${\rm (i)}$ can be proved as for C*-algebras \cite[Section 2.3.3]{brarob} with the only care of replacing $\|\omega\|$ with $\omega(e)$. We now prove the equivalence of (ii) and (iii).
\\${\rm (ii)} \Rightarrow  {\rm (iii)}$. If $\pi_\omega$ is not irreducible, then $\pi_\omega(\A)'_\w= \pi(\Ao)'_\w$, which is a nontrivial von Neumann algebra, contains a nontrivial projection $P$. Define
$$\nu(a)=\ip{\pi_\omega (a) \xi_\omega}{P\xi_\omega}, \quad a\in \A.$$ Then, as it is easily seen, $\nu$ is a representable functional on $(\A, \Ao)$. Moreover, if $a \in \A^+$, we have
$$\omega(a)-\nu(a)= \ip{\pi_\omega (a) \xi_\omega}{(I-P)\xi_\omega}=\ip{\pi_\omega (a)(I-P) \xi_\omega}{(I-P)\xi_\omega} \geq 0.$$
Hence, $\nu\leq \omega$ and $\nu$ is not a multiple of $\omega$.

${\rm (iii)} \Rightarrow  {\rm (ii)}$. If $\pi_\omega$ is quasi-irreducible, then $\pi_\omega(\A)'_\w=\pi_\omega (\Ao)'_\w ={\mb C}I$. Let $\nu\in
\B(\A)$, $ \mbox{ $0\leq\nu\leq\omega$}$. Since $\omega_0:=\omega \upharpoonright_{\Ao}$ and the (usual) GNS representation $\rho_{\omega_O}$ is unitarily equivalent to  $\pi_\omega \upharpoonright_{\Ao}$ (or, more precisely, to the representation obtained by extending every operator $\pi_\omega(a)$, $a \in \Ao$ to the Hilbert space $\Hil_{\pi_\omega}$ completion of $\D_{\pi_\omega}$). This implies that $\rho_{\omega_O}(\Ao)'=\pi_\omega (\Ao)'_\w ={\mb C}I$. Hence, $\omega_0$ is pure on $\Ao$; so $\nu_0:= \nu \upharpoonright_{\Ao}$ is a multiple of $\omega_0$; i.e. $\nu_0=\lambda \omega_0$, for some $\lambda\in [0,1]$. The continuity of the states and the density of $\Ao$ in $\A$ w. r. to the topology ${\tau_{rep}}$ easily imply that $\nu=\lambda \omega$; i.e. $\omega$ is pure.

\end{proof}

\section{Asymptotic behavior}\label{sect_two}

In \cite{bit0} Inoue and two of us studied local modifications of states on a quasi-local quasi *-algebra. We give here some more properties coming from this notion.

We define the \emph{local modification}, $\omega_b,$ of an arbitrary state
$\omega\in \E(\A)$, due to the action of an element, $b$, of $\A_0$ by the
formula
$$\omega_b(a)=\tfrac{\omega(b^*ab)}{\omega(b^*b)}.$$

In what follows we will always assume that $\omega(b^*b)\neq0$.
It is possible to check that conditions (L1)-(L3) are stable under
the map $\omega\rightarrow \omega_b$, with $b\in\Ao$ \cite{bit0}.
Therefore for every $b$ of $\A_0$ and $\omega \in \E(\A)$ the
modification $\omega_b$ also belongs to $\E(\A)$ and it is $\tau_{rep}$ continuous.

\medskip
The following definition selects states on $\A$ with a {\em
reasonable asymptotic behavior}. These states, indeed, factorize
on regions far enough from the support of a given element.

\bedefi\label{defalclu}
A state $\omega$ over $\A$ is said to be almost clustering (AC)
if, $\forall \,b\in\Ao$ and $\forall\,\epsilon>0$, there exists
$\alpha\in\F$, $\alpha\geq \alpha_b$, such that,
$\forall\gamma\perp\alpha$ we have
$\left|\omega(ab)-\omega(a)\omega(b)\right|\leq \epsilon \|a\|
\|b\| $, $\forall\,a\in\A_\gamma$.
\findefi
Similar definitions are given in many textbooks, like \cite{sew},
\cite{brarob} and \cite{emch}, where the physical motivations are
discussed in  detail.

\begin{prop} \label{16} If a state $\omega$ over $\A$ is AC
 then, for every $b \in \Ao$, the modification $\omega_b$ is almost clustering.
\end{prop}
\begin{proof}
If a state $\omega$ over $\A$ is a AC, $\forall
\,b\in\Ao$ and $\forall\,\epsilon>0$, there exists $\alpha\in\F$,
$\alpha\geq \alpha_b$ such that, $\forall\gamma\perp\alpha$ we
have $\left|\omega(ab)-\omega(a)\omega(b)\right|\leq \epsilon
\|a\|\|b\|$, $\forall\,a\in\A_\gamma$.

For every $c\in\Ao$ there exist $\alpha_c$ the support of $c$ and a
$\F\ni\alpha_{c,b}\geq\alpha_c,\alpha_b$ ($\F$ is a set of indexes
directed upward). For every $\gamma\bot\alpha_{c,b} $ and
$a\in\A_\gamma$
\begin{align*}
&\left|\omega_c(ab)-\omega_c(a)\omega_c(b)\right|=\left|\frac{\omega(c^*abc)\omega(c^*c)-\omega(c^*ac)\omega(c^*bc)}{\omega(c^*c)^2}\right|\\
&=
\left|\frac{\omega(ac^*bc)\omega(c^*c)-\omega(a)\omega(c^*bc)\omega(c^*c)+
\omega(a)\omega(c^*bc)\omega(c^*c)-\omega(c^*ac)\omega(c^*bc)}{\omega(c^*c)^2}\right| \\
&\leq
\frac{\mid\omega(ac^*bc)-\omega(a)\omega(c^*bc)\mid\omega(c^*c)+
\mid\omega(c^*ac)-\omega(a)\omega(c^*c)\mid \,
\mid\omega(c^*bc)\mid }{\omega(c^*c)^2} \\&\leq
\frac{\epsilon\|c\|^2
\parallel a \parallel \|b\| \omega(c^*c)+\mid \omega(ac^*c)-\omega(a)\omega(c^*c)\mid  \omega(c^*bc)\mid }{\omega(c^*c)^2}
\\&\leq
\frac{\epsilon\|c\|^2
\parallel a \parallel \|b\| \omega(c^*c)+\epsilon\|c\|^2
\parallel a \parallel \mid\omega(c^*bc)\mid }{\omega(c^*c)^2}\\
&=\epsilon\|c\|^2
\parallel a \parallel \frac{ \|b\|\omega(c^*c)+\mid\omega(c^*bc)\mid
}{\omega(c^*c)^2}\leq \epsilon \|c\|^2
\parallel a \parallel \|b\|\frac{ \omega(c^*c)+\omega(c^*c)}{\omega(c^*c)^2}\\&=2\epsilon \|c\|^2
\parallel a \parallel \parallel b\parallel \frac{ 1}{\omega(c^*c)}.
   \end{align*}

\end{proof}

Let us now consider a group $G$ of invertible maps from $\F$ onto $\F$ (think of translations, in the standard case).
To every element $g \in G$ it corresponds an isometric *-automorphism $\tau_g$ of $\Ao$
satisfying the
following covariance relation: $\tau_g(\A_\alpha)=\A_{g(\alpha)}$. The map $\tau: g \in G \to \tau_g \in Aut(\Ao)$ is a representation of $G$.

A state $\omega$ is invariant under $G$ (shortly, $G$-invariant) if
$\omega(\tau_g(x))=\omega(x)$ for all $g\in G$ and $x\in\A_0$.

Let us now fix a sequence ${\mathbf g}:=\{g_j, j\in{\Bbb{N}}\}$  of elements of $G$ with the property that, for all $\alpha\in\F$,  there exists $j_\alpha\in {\mathbb N}$ such that, for all $j\geq j_\alpha$,   $g_j(\alpha)\perp\alpha$.
\berem In a concrete realization, this means
that the {local} element $x$ is moved  {\em towards infinity} by the action of the group.
\enrem

Given $x\in\A_0$ we can define a new element of $\A_0$,  $x_N$,
$N\geq 1$, via the following mean: \be
x_N=\frac{1}{N}\,\sum_{j=1}^N\,\tau_{g_j}(x), \quad N \in {\Bbb{N}}. \label{31}\en

Consider first the case where $\tau_{rep}-\lim_{N\to \infty}x_N$ exists in $\A$ and let us call $x_\infty$ this limit. Then, for every state $\omega\in \E(\A)$, we have
$$\omega (x_\infty)=\lim_{N\to \infty}\omega(x_N),$$
since the GNS representation corresponding to $\omega$ is $\tau_{rep}$-continuous.

More in general, we put
\begin{align*}  &D(x_\infty)=\{ \omega \in \E(\A); \lim_{N\to \infty}\omega(x_N) \mbox{ exists }\};\\ &\omega(x_\infty):=\lim_{N,\infty}\,\omega(x_N), \,\omega \in D(x_\infty).
\end{align*}

\berem We warn the reader about the fact that this doesn't define, in general, an element $x_\infty$ of $\A$. But sometimes it does. For instance,
if $\omega$ is invariant under
G, then $\omega\in D(x_\infty)$, for all $x\in\A_0$, and,
furthermore, $\omega(x_\infty)=\omega (x)$. Hence, if the set
of all the states which are invariant under G is
sufficient, this would imply that an element $x_\infty \in \Ao$ exists for all $x\in\A_0$, and that $x=x_\infty$.

\enrem

We have already seen that the modification $\omega_b$ of a given
state $\omega$, $b\in\A_0$ and
$\omega_b(.)=\frac{\omega(b^*.b)}{\omega(b^*b)}$, shares with
$\omega$ itself some important properties: first if $\omega$
satisfies conditions (L1)-(L3), then $\omega_b$ satisfies the same
conditions. Moreover, if $\omega$ is AC, then $\omega_b$ is AC as
well. To this list we can add the following result, which is close
to what discussed in \cite{sew} in a standard context:

\begin{thm}
Let $x\in\A_0$ and $\omega\in D(x_\infty)$ an AC state. Hence,
for all $b\in\A_0$ we have $\omega_b\in D(x_\infty)$ and
$\omega_b(x_\infty)=\omega(x_\infty)$.
\end{thm}

\begin{proof}
We begin the proof with the following remark: due to our local
structure, it is clear that for every $b\in\Ao$ there exists an index $M$, $1\leq M\leq N$, such that,
for $1\leq j\leq M$, $\tau_{g_j}(x)b$ is in general different from
$b\tau_{g_j}(x)$, while they coincide for $j>M$. In other words, $\tau_{g_j}(x)$ and $b$ commute when they are {\em distant enough}. Moreover, since
$\omega$ is AC, we can also assume that for all $\epsilon>0$, and for  $j$ large enough, the following
inequality holds
$$\left|\omega(b^*b\tau_{g_j}(x))-\omega(b^*b)\omega(\tau_{g_j}(x))\right|\leq\epsilon\|b\|^2\|x\|.$$
Hence we can
check that
$$
\left|\omega_b(x_N)-\omega(x_\infty)\right|\rightarrow 0
$$
when $N\rightarrow\infty$. This in particular implies that
$\omega_b$ belongs to $D(x_\infty)$ and, more than this, that
$\omega_b(x_\infty)=\omega(x_\infty)$.
\end{proof}

\begin{cor}
Let $x\in\A_0$ and $\omega\in D(x_\infty)$ an AC state. Let
$\{b_j\in\A_0,\,j=1,2,\ldots,K\}$ be a set of elements of $\A_0$
and $\{\lambda_j\geq0,\,j=1,2,\ldots,K\}$ with
$\sum_{j=1}^K\,\lambda_j=1$. Then the state
$\phi:=\sum_{j=1}^K\,\lambda_j\omega_{b_j}$ belongs to
$D(x_\infty)$ and $\phi(x_\infty)=\omega(x_\infty)$.
\end{cor}

The proof is straightforward and will be omitted.

\berem
  Let $x$, $\omega$ and $\phi$ be as in  the previous Corollary. If the set of the $\phi$'s is sufficient, we can again conclude that an element $x_\infty$ does exist in $\Ao$, and it coincides with $x$.
\enrem

\bedefi
A state $\omega$ over $\A$ is termed primary if
$\pi_{\omega}(\A)'_\w\cap(\pi_{\omega}(\A)'_\w)'$ is trivial; i.e.
it consists only of the multiples of the unit operator.
\findefi

\bedefi A state $\omega$ over $\A$ has the cluster property (relatively to $\mathbf g$) if
$$\mid\omega(a\tau_{g_j}(x))-\omega(a)\omega(\tau_{g_j}(x))\mid \rightarrow 0, \quad a\in\A,\, x\in\Ao,$$
when $j\to \infty$.
\findefi

\begin{prop}
If $\omega$ has the cluster property then, for all $b\in\Ao$, $\omega_b$ has the cluster property.
\end{prop}

\begin{prop}
Let $(\A,\Ao)$ be a {\em quasi-local quasi *-algebra}, $x\in \Ao$ and
$\omega\in D(x_\infty)$. If
$\omega$ is  primary  then it has the cluster property and the following holds:
\begin{equation} \label{asym}\mid\omega(ax_N)-\omega(a)\omega(x_\infty)\mid \rightarrow 0, \quad a\in\A,\end{equation}
when $N\rightarrow\infty$.

\end{prop}

\begin{proof} The cluster property can be proved as in \cite[Theorem 3.2.2]{haag}, taking into account the equality $\pi_{\omega}(\A)'_\w=\pi_{\omega}(\Ao)'_\w$ stated in Proposition \ref{nove}.
The proof of \eqref{asym} is based in very standard estimates that we omit.
\end{proof}

%\subsection{Some remarks on sesquilinear forms}

\berem As discussed at length in \cite{ait_book} the role that in the theory of representations of *-algebras is played positive linear functional is more conveniently covered, when passing to partial *-algebras, by certain sesquilinear forms enjoying some {\em invariance} property. The interplay of these two notions has been studied in \cite{ct_ban} and \cite{bit}. Some results given in the previous discussion, can be easily extended to sesquilinear form.

We denote with $\mc T(\A) $ the set of all sesquilinear forms
$\Omega$ on $\A \times \A$ with the following properties
\begin{itemize}
\item[(i)]$ \Omega(a,a) \geq 0 \hspace{3mm} \forall a \in \A$;
\item[(ii)] $\Omega(xa,b) = \Omega (a,x^*b),  \quad \forall a,b\in\A,\hspace{2mm} \forall x\in \Ao$.

%\item[(iv)] .... condizione i.p.s.(?) .
\end{itemize}

By (ii) it follows that the linear functional $\vp_a$ defined, for $a \in \A$, by $\vp_a(x)=\Omega(xa,a)$ is positive  and, hence, bounded. Thus, the following inequality holds:
\be
 |\Omega(x\,a,a)|\leq \parallel x \parallel \Omega(a,a),  \quad \forall x\in\A_0,\hspace{2mm} \forall a\in \A.
 \label{51}
 \en

The \emph{modification} $\Omega_b$ of an arbitrary
sesquilinear form $\Omega\in \mc S(\A)$, due to the action of an
element $b$ of $\A_0$, such that $\Omega(b,b)>0$, is then defined by the formula
$$\Omega_b(x,x)=\tfrac{\Omega(xb,xb)}{\Omega(b,b)}.$$

It is easy to check that conditions (i)-(ii) are stable under
the map $\Omega\rightarrow \Omega_b$, with $b\in\Ao.$ Then, since $\Omega$ satisfies (\ref{51}), $\Omega_b$ satisfies an analogous estimate.

Therefore, for every $b$ of $\A_0$ and $\Omega \in \mc T(\A)$, the
modification $\Omega_b$ also belongs to $\mc T(\A).$
\medskip

We now extend Definition \ref{defalclu} to the present settings:
A sesquilinear forms $\Omega$ over $\A$ is said to be
AC if, $\forall \,b\in\Ao$ and
$\forall\,\epsilon>0$, there exists $\alpha\in\F$, $\alpha\geq
\alpha_b$, such that, $\forall\gamma\perp\alpha$ we have
$\left|\Omega(a,b)-\Omega(a,e)\Omega(e,b)\right|\leq \epsilon
\|a\| \|b\| $, $\forall\,a\in\A_\gamma$.
Then, with minor modifications of the proof of Proposition \ref{16} one can prove that the property AC is preserved by local modifications; i.e. if $\Omega$ is AC, then, for every $c \in \Ao$, the modification $\Omega_c$ is
AC.

Other results related to the states could be also restated in terms of sesquilinear form with just minor modifications. \enrem

\section*{Acknowledgements}

This paper was partially supported by MURST. FB acknowledges financial support by the project {\em Problemi Matematici Non Lineari di
Propagazione e Stabilit\`a nei Modelli del Continuo}, coordinated by
Prof. T. Ruggeri.

\end{document}